\documentclass[conference]{IEEEtran}

\usepackage{stfloats}
\usepackage{amsfonts}
\usepackage{amssymb}
\usepackage{amsthm}
\usepackage{cite}
\usepackage[cmex10]{amsmath}
\usepackage{float}
\usepackage{color}
\usepackage{stfloats,fancyhdr}
\usepackage{amsmath,bm}
\usepackage{algorithm}
\usepackage{algorithmic}
\usepackage{multirow}
\usepackage{changepage}
\usepackage[normalem]{ulem}
\usepackage{amsthm}
\usepackage{balance}

\newtheorem{theorem}{Theorem}

\newtheorem{proposition}{Proposition}

\newtheorem{assumption}{Assumption}
\newtheorem{remark}{Remark}

\ifCLASSINFOpdf
\usepackage[pdftex]{graphicx}
\DeclareGraphicsExtensions{.pdf,.jpeg,.png}
\else
\usepackage[dvips]{graphicx}
\DeclareGraphicsExtensions{.eps}
\fi

\usepackage{subfigure}
\usepackage{fancybox,dashbox}
\usepackage{authblk}

\begin{document}
		
\title{To Retransmit or Not: Real-Time Remote Estimation in Wireless Networked Control}

\author{Kang Huang, Wanchun Liu$^\dagger$, Yonghui Li, and Branka Vucetic\\
	School of Electrical and Information Engineering, The University of 
	Sydney, Australia\\
	Emails:	\{kang.huang,\ wanchun.liu,\ yonghui.li,\ branka.vucetic\}@sydney.edu.au. 
}



\maketitle

\begin{abstract}
\let\thefootnote\relax\footnote{$^\dagger$\emph{W. Liu is the corresponding author.}
}	
	Real-time remote estimation is critical for mission-critical applications including industrial automation, smart grid, and the tactile Internet.
	In this paper, we propose a hybrid automatic repeat request (HARQ)-based real-time remote estimation framework for linear time-invariant (LTI) dynamic systems.
	Considering the estimation quality of such a system, there is a fundamental tradeoff between the reliability and freshness of the sensor's measurement transmission.
	When a failed transmission occurs, the sensor can either retransmit the previous old measurement such that the receiver can obtain a more reliable old measurement, or transmit a new but less reliable measurement.
	To design the optimal decision, we formulate a new problem to optimize the sensor's online decision policy, i.e., to retransmit or not, depending on both the current estimation quality of the remote estimator and the current number of retransmissions of the sensor, so as to minimize the long-term remote estimation mean-squared error (MSE).
	This problem is non-trivial.
	In particular, it is not clear what the condition is in terms of the communication channel quality and the LTI system parameters, to ensure that the long-term estimation MSE can be bounded.
	We give a sufficient condition of the existence of a stationary and deterministic optimal policy that stabilizes the remote estimation system and minimizes the MSE.
	Also, we prove that the optimal policy has a switching structure, and derive a low-complexity suboptimal policy.
	Our numerical results show that the proposed optimal policy notably improves the performance of the remote estimation system compared to the conventional non-HARQ policy.

\end{abstract}

	\section{Introduction}	
	Real-time remote estimation is critical for networked control applications such as
	industrial automation, smart grid, vehicle platooning, drone swarming, immersive virtual reality (VR) and the tactile Internet~\cite{Antonakoglou2018towards}.
	For such real-time applications, high-quality remote estimation of the states of dynamic processes over unreliable links is a major challenge.
	The sensor's sampling policy, the estimation scheme at a remote receiver, and the communication protocol for state-information delivery between the sensor and the receiver should be designed jointly.

	
	To enable optimal design of wireless {remote estimation}, the performance metric for the remote estimation system needs to be selected properly. 
	For some applications, the model of the dynamic process under monitoring is unknown and the receiver is not able to estimate the current state of the process based on the previously received states, i.e., a state-monitoring-only scenario~\cite{kaul2012real}. In this scenario, the performance metric is the age-of-information (AoI), which reflects how old the freshest received sensor measurement is, since the moment that measurement was generated at the sensor~\cite{kaul2012real}.
	However, in practice, most of the dynamic processes are time-correlated, and the state-changing rules can be known by the receiver to some extent. 
	Therefore, the receiver can estimate the current state of the process based on the previously received measurements and the model of the dynamic process (see e.g., \cite{schenato2008optimal,sun2017remote}), especially when the packet that carries the current sensor measurement is failed or delayed. 
	In this sense, the estimation mean-squared error (MSE) is the perfect performance metric.	
	

	From a communication protocol design perspective, we naturally ask: does a sensor need retransmission or not for mission-critical real-time remote estimation?
	Retransmission is required by conventional communication systems with non-real-time backlogged data to be perfectly delivered to the receivers.
	Also, energy-constrained remote estimation systems and the ones with low sampling rate can also benefit from retransmissions, see e.g.,~\cite{liu2016energy} and~\cite{Demirel2015to}. 
	It seems that retransmissions may not improve the performance of a mission-critical real-time remote estimation system~\cite{gupta2010estimation}, which is not mainly constrained by energy nor sampling rate,
	as it is a waste of transmission opportunity to transmit an out-of-date measurement instead of the current one.
	However, this is true only when a retransmission has the same success probability as a new transmission, e.g., with the standard automatic repeat request (ARQ) protocol.
	Note that a hybrid ARQ (HARQ) protocol, e.g., with a chase combining or incremental redundancy scheme, is able to effectively increase the successful detection probability of a retransmission by combining multiple copies from previously failed transmissions~\cite{caire2001throughput}.
	Therefore, a HARQ protocol has the potential to improve the performance of real-time remote estimation.
	However, to the best of our knowledge, HARQ has never been considered in the open literature of real-time remote estimation of a time-correlated dynamic process.
	

	In the paper, we introduce HARQ into real-time remote estimation systems and optimally design the sensor's transmission policy to minimize the estimation MSE.
	Note that there is a fundamental tradeoff between the reliability and freshness of the sensor's measurement transmission. When a failed transmission occurs, the sensor can either retransmit the previous old measurement such that the receiver can obtain a more reliable old measurement, or transmit a new but less reliable measurement.
	The main contributions of the paper are summarized as follows:
	\begin{itemize}
		\item We propose a novel HARQ-based real-time remote estimation system, where the sensor makes online decision to send a new measurement or retransmit the previously failed~one depending on both the current estimation quality of the receiver and the current number of retransmissions of the sensor.
		\item We formulate the problem to optimize the sensor's decision policy so as to maximize the long-term performance of the receiver in terms of the average MSE.
		Since it is not clear whether the long-term average MSE can be bounded or not,
		we give a sufficient condition in terms of the communication channel quality and the LTI system parameters
		to ensure that an optimal policy exists and stabilizes the remote estimation system.		
		\item We derive a structural property of the optimal policy, i.e., the optimal policy is a switching-type policy, and give an easy-to-compute suboptimal policy.
		Our numerical results show that the suboptimal policy can efficiently improve the system performance than the conventional non-HARQ policy, under the setting of practical system parameters.
		
	\end{itemize}


\section{System Model} \label{section2}
We consider a basic system setting that a \emph{smart sensor} periodically samples, pre-estimates and sends its local estimation of a dynamic process to a remote receiver through a wireless link with packet dropouts, as illustrated in Fig.~\ref{retransmit_system_model}.

\subsection{Dynamic Process Modeling}
   We consider a general discrete linear time-invariant (LTI) model for the dynamic process as (see e.g., \cite{yang2015deterministic,shi2012optimal,yang2013schedule})
   \begin{equation} \label{sys}
	\begin{aligned}
	   x_{k+1} &= {A}x_k + w_k,\\
	   y_k &= {C}x_k + v_k,
	\end{aligned}
   \end{equation}
     where 
     the discrete time steps are determined by the sensor’s sampling period $T_s$,
     $x_k \in \mathbb{R}^n$ is the process state vector, $A \in \mathbb{R}^{n \times n}$ is the state transition matrix, $y_k \in \mathbb{R}^m$ is the measurement vector of the smart sensor attached to the process, $C \in \mathbb{R}^{m \times n}$ is the measurement matrix\footnote{Note that $C$ is not necessary to be full rank~\cite{maybeck1979stochastic}, as illustrated in Fig.~\ref{sys}, i.e., $x_k$ is a two-dimensional (2D) signal, while the measurement $y_k$ is one-dimensional. After Kalman filtering, we have a 2D $\hat{x}^s_k$.}, $w_k \in \mathbb{R}^n$ and $v_k \in \mathbb{R}^m$ are the process and measurement noise vectors, respectively. We assume $w_k$ and $v_k$ are independent and are identically distributed (i.i.d.) zero-mean Gaussian processes with corresponding covariance matrices $Q$ and $R$, respectively. The initial state $x_0$ is zero-mean Gaussian with covariance matrix $\Sigma_0$.     
     To avoid trivial problems, we assume that $\rho^2(A) > 1$, where $\rho^2(A)$ is the maximum squared eigenvalue of $A$~\cite{shi2012scheduling}. 

     \begin{figure}[t]
      \centering\includegraphics[scale=0.6]{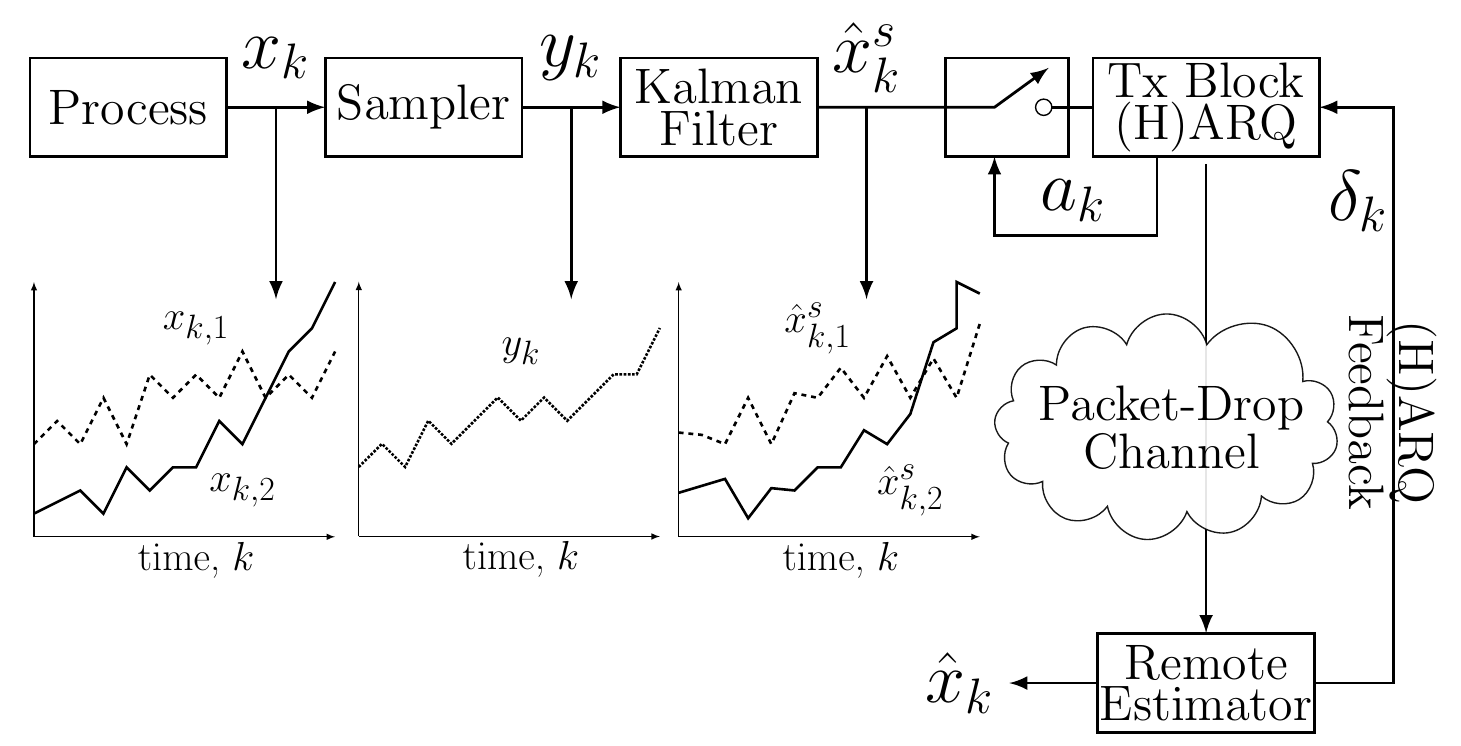}
     	\caption{Proposed remote estimation system with HARQ, where $x_k\triangleq \left[x_{k,1},x_{k,2}\right]^T$ is the two-dimensional state vector of the dynamic process.}
     	\label{retransmit_system_model}
     	\vspace{-0.5cm}
     \end{figure}

\subsection{State Estimation at the Smart Sensor}
Since the sensor's measurements are noisy, the smart sensor with sufficient computation and storage capacity is required to estimate the state of the process, $x_k$, using a Kalman filter~\cite{shi2012optimal,yang2013schedule}, which gives the minimum estimation MSE, based on the current and previous raw measurements:
%
%
%
%
\begin{subequations}\label{sub:1}
\begin{align}
x_{k|k-1}^s&=Ax_{k-1|k-1}^s\\
P_{k|k-1}^s&=AP_{k-1|k-1}^sA^T+Q\\
K_k&=P_{k|k-1}^sC^T(CP_{k|k-1}^sC^T+R)^{-1}\\
x_{k|k}^s&=x_{k|k-1}^s+K_k(y_{k}-Cx_{k|k-1}^s)\\
P_{k|k}^s&=(I-K_{k}C)P_{k|k-1}^s
\end{align}
\end{subequations}
where $I$ is the $m \times m$ identity matrix, $(\cdot)^T$ is the transpose operator, $x^s_{k|k-1}$ is the priori state estimation, $x^s_{k|k}$ is the posteriori state estimation at time $k$, $K_k$ is the Kalman gain, $P_{k|k-1}$ and $P_{k|k}$ represent the priori and posterior error covariance at time $k$, respectively. The first two equations present the prediction steps while the last three equations present the updating steps~\cite{maybeck1979stochastic}.
Note that $x^s_{k|k}$ is the output of the Kalman filter at time $k$, i.e., the pre-filtered measurement of~$y_k$, with the estimation error covariance $P_{k|k}^s$.

As we focus on the effect of communication protocols on the stability and quality of the remote estimation, we assume that the local estimation is stable as follows~\cite{shi2012optimal,yang2013schedule}.
\begin{assumption}
	\normalfont
	The local Kalman filter of system \eqref{sys} is stable with the system parameters $\{A, C,Q\}$\footnote{The rigorous stability condition in terms of $\{A, C,Q\}$ is given in~\cite{maybeck1979stochastic}.}, i.e.,
	the error covariance matrix $P_{k|k}^s$ converges to a finite matrix~$\bar{P}_0$ when $k$ is sufficiently large.
\end{assumption}

\emph{In the rest of the paper, we assume that the local Kalman filter operates in the steady state~\cite{shi2012optimal,yang2013schedule}, i.e., $P_{k|k}^s = \bar{P}_0$. For ease of notation, we use $\hat{x}_k^s$ to denote the sensor's estimation, $x_{k|k}^s$.}
       
\subsection{Communication and Remote Estimation}   \label{sec:ARQ} 
     The sensor transmits its pre-filtered measurement in a packet and sends it to the receiver (i.e., the remote estimator) through a static channel, which is modeled as an i.i.d. packet-dropping process.\footnote{General fading and Markov channels can be considered in our future work.}
     Let $(1-\lambda)$ denotes the packet-drop probability.
     Note that the \emph{successful packet detection probability} at the receiver can be different for different transmission/retransmission schemes.
     
     \emph{We assume that the packet length is equal to the sampling period $T_s$. Thus, there exists a unit transmission delay between the sensor and the receiver.} 
     For example, the sensor's raw measurement at the beginning of time slot~$k$ is filtered and sent to the receiver before time slot $(k+1)$.     
     Also, we assume that the acknowledgement/negative-acknowledgement
     (ACK/NACK) message is fed back from the receiver to the sensor perfectly without any delay, when the packet detection succeeds/fails.
     If an ACK is received by the sensor, it will send a new (pre-filtered) measurement in the next time slot. If a NACK is received, the sensor may decide whether to retransmit the unsuccessfully transmitted measurement based on its ARQ protocol or to send the new measurement.
    In the rest of this section, we introduce the standard ARQ-based estimation system. The proposed HARQ-based protocol will be presented in Sec.~\ref{sec:HARQ}.

\emph{\underline{Standard ARQ-Based Remote Estimation.}}	
For the standard ARQ protocol, the receiver discards the failed packets, and the sensor simply resends the previously failed packet if a retransmission is required.
Thus, the successful packet detection probability~at each time is independent of the current number of retransmissions.
Let the random variable $\delta^{\text{ARQ}}_k  \in \{0,1\}$ denote the failed/successful packet detection  at the receiver  in time slot $k$. We have
\begin{equation}
\mathbb{P}\left[\delta^{\text{ARQ}}_k=1\right]=\lambda, \forall k.
\end{equation}
As the chances of the successful detection of a new transmission and a retransmission are the same, 
\emph{the optimal policy is to always transmit the current sensor estimation, i.e., a non-retransmission policy~\cite{gupta2010estimation}.}

Consider the non-retransmission policy.
As the successfully detected packet contains the estimated state information with a \emph{one-step delay}, the receiver needs to estimate the current state based on the dynamic process model~\eqref{sys}. If the packet detection is failed, the receiver can estimate the current state based on its previous estimation and the process model.
Therefore, the optimal estimator at the receiver is given~as~\cite{schenato2008optimal} 
\begin{equation}
\hat{x}_k = 
\begin{cases}
A\hat{x}_{k-1}^s, &\mbox{if } \delta^{\text{ARQ}}_{k-1} = 1 \\
A\hat{x}_{k-1}, &\mbox{otherwise.}
\end{cases}
\end{equation}

\section{HARQ-Based Remote Estimation}  \label{sec:HARQ}
For a HARQ protocol, the receiver buffers the incorrectly received packets,
and the detection of the retransmitted packet depends on all the buffered related packets.\footnote{To be specific, if a retransmission is required, the sensor can either resend the previously failed packet (i.e., a chase combining scheme) or send a retransmission packet that contains different information than the previous one (i.e., a incremental redundancy scheme).
The receiver is possible to successfully detect the current retransmission packet based on the previously erroneously received ones~\cite{frenger2001performance,tripathi2003reliability}.} Thus, the probability of successful packet detection in time slot $k$, depends on the number of consecutive retransmissions $r_k \geq 0$~\cite{tripathi2003reliability}. In particular, $r_k = 0$ indicates a new transmission in time slot $k$.

Let the random variable $\delta^{\text{HARQ}}_k  \in \{0,1\}$ denote the failed/successful packet detection at the receiver in time slot~$k$.
Thus, the successful packet detection probability is given as~\cite{tripathi2003reliability}
\begin{equation} \label{HARQ_prob}
\mathbb{P}\left[\delta^{\text{HARQ}}_k=1\right]= 1 - g(r_k), \forall k,
\end{equation}
where the function $g(\cdot)$ is determined by the specific HARQ protocol (e.g., with chase combining or incremental redundancy). \emph{Specifically, $1- g(0) = \lambda$ and $g(0)>g(r)$ when $r>0$, i.e., a retransmission is more reliable than a new transmission.}

In this scenario, when a failed transmission occurs, there exists an inherent trade-off between retransmitting previously failed local state estimation with a higher success probability, and sending the current state estimation with a lower success probability.
Therefore, the sensor needs to properly decide when to transmit a new estimation and when to retransmit.

Let $a_k \in \{0,1\}$ be the sensor's decision variable at time~$k$, as illustrated in Fig.~\ref{sys}. If $a_k=0$, the sensor sends the new measurement to the receiver in time slot $k$; otherwise, it retransmits the  unsuccessfully transmitted measurement.
Thus, the current number of retransmissions, $r_k$, has the update rule~as
\begin{equation} \label{retransmission time}
r_k =\begin{cases}
0, &\mbox{if } a_{k} = 0 \\
r_{k-1}+1, &\mbox{otherwise.} 
\end{cases}
\end{equation}

If a packet transmitted in time slot $(k-1)$ is successfully detected, the receiver can estimate the current state $x_k$ based on the received sensor's estimation at $(k-1-r_{k-1})$ as illustrated in Fig.~\ref{fig:packet_process},
and the system dynamics~\eqref{sys}. Otherwise, the receiver can only do estimation based on the previous one. Thus, the receiver estimator based on HARQ is given as
\begin{equation} \label{HARQ_estimator}
\hat{x}_k =
\begin{cases}
A\hat{x}_{k-1}^s, &\mbox{if } a_{k-1}=0 \mbox{ and } \delta^{\text{HARQ}}_{k-1}=1\\ 
A^{r_{k-1}+1}\hat{x}_{k-r_{k-1}-1}^s, &\mbox{if } a_{k-1}=1 \mbox{ and }  \delta^{\text{HARQ}}_{k-1}=1\\
A\hat{x}_{k-1}, &\mbox{otherwise.}
\end{cases}
\end{equation}
From the second expression of \eqref{HARQ_estimator}, the estimation quality of $x_k$ is not good if $r_{k-1}$ is large, since the receiver's current estimation $\hat{x}_k$ is based on the sensor's measurement at time $(k-r_{k-1}-1)$, i.e., an out-of-date information.
From the last expression of \eqref{HARQ_estimator}, the estimation quality of $x_k$ is bad if there is a sequence of failed transmissions and the receiver estimates the current state based on the one sent by the sensor a long time ago.

          \begin{figure}[t]
	\centering\includegraphics[scale=0.8]{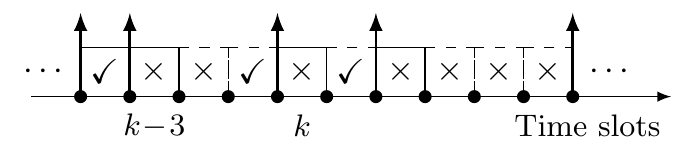}
	\vspace{-0.5cm}
	\caption{An illustration of the sensor's transmission process. The solid circles denote the raw measurement sampling time, 
		the up arrows are the starting points of new transmissions (i.e., only these (pre-filtered) measurements will be sent to the receiver), solid/dashed blocks are new/re-transmission packets, and $\checkmark$/$\times$ denotes a successful/failed detection at the receiver.}
	\label{fig:packet_process}
	\vspace{-0.5cm}
\end{figure} 

For ease of analysis, we define the estimation quality index, $q_k$, as
\begin{equation}\label{def_q_0}
q_k \triangleq k-t_k,
\end{equation}
where $t_k$ is the generation time slot of the latest sensor's estimation that is successfully received by the receiver before time slot $(k+1)$,\footnote{Note that the definition of $q_k$ is similar to that of AoI~\cite{kaul2012real}, which will be further discussed in Sec.~\ref{sec:delay_optimal}.} and $q_k\geq 0$.
As it is straightforward that
$
t_k = 
\begin{cases}
k- r_k, &\mbox{ if } \delta^{\text{HARQ}}_{k}=1\\
t_{k-1}, &\mbox{ otherwise}
\end{cases}
$,
we have
\begin{equation} \label{def_q}
q_k = 
\begin{cases}
r_k, &\mbox{ if } \delta^{\text{HARQ}}_{k}=1\\
q_{k-1}+1, &\mbox{ otherwise.}
\end{cases}
\end{equation}
Therefore, the last iteration expression in \eqref{HARQ_estimator} can be further written as
\begin{equation}\label{HARQ_estimator_2}
\hat{x}_k = A^{q_{k-1}+1} \hat{x}^s_{k-q_{k-1} -1}, \mbox{ if } \delta^{\text{HARQ}}_{k-1}=0.
\end{equation}
In other words, the receiver estimation at time $k$ is based on the state estimation of the smart sensor at time $(k-q_{k-1}-1)$.

Therefore, from \eqref{HARQ_estimator}, \eqref{def_q} and \eqref{HARQ_estimator_2}, the estimation error covariance can be obtained~as
\begin{align} \label{covariance1}
P_k 
&\triangleq  \mathbb{E}\left[(x_k-\hat{x}_k)(x_k-\hat{x}_k)^T\right]\\
\label{covariance2}
&=
\begin{cases}
f(\bar{P}_0), &\mbox{if } a_{k-1}=0 \mbox{ and } \delta^{\text{HARQ}}_{k-1}=1\\
f^{r_{k-1}+1}(\bar{P}_0), &\mbox{if } a_{k-1}=1 \mbox{ and }  \delta^{\text{HARQ}}_{k-1}=1\\
f^{q_{k-1}+1}(\bar{P}_0), &\mbox{otherwise}
\end{cases}\\ 
&=f^{q_{k-1}+1}(\bar{P}_0) \label{general_form}
\end{align}
where \eqref{general_form} is obtained by taking \eqref{retransmission time} and \eqref{def_q} into \eqref{covariance2},
 $f(X)\triangleq AXA^T+Q$, $f^{n+1}(\cdot)  \triangleq f (f^{n}(\cdot))$ when $n\geq 1$, and $f^{1}(\cdot) \triangleq f(\cdot)$.	
Note that $P_k$ takes value from a countable infinity set, i.e., $P_k \in \{f(\bar{P}_0),f^2(\bar{P}_0),\cdots\}$.	
The operator $f^n(\bar{P}_0)$ is monotonic with respect to (w.r.t.) $n$, i.e.,   the matrix $f^{n_1}(\bar{P}_0) \leq f^{n_2}(\bar{P}_0)$ in element wise if $1 \leq n_1 \leq n_2$,  and hence $\text{Tr}\left(f^{n_1}(\bar{P}_0)\right) \leq \text{Tr}\left(f^{n_2}(\bar{P}_0)\right)$, where $\text{Tr}\left(\cdot \right)$ is the trace operator (see Lemma 3.1 in \cite{shi2012scheduling}).

\emph{\underline{Performance Metric and Problem Formulation.}}
Based on the estimation error covariance $P_k$ in \eqref{covariance1}, the estimation MSE of $x_k$ is $\text{Tr}\left(P_k\right)$. 
Thus, the long-term average MSE of the dynamic process is defined as
\begin{equation}\label{cost_function}
\limsup_{K\to\infty}\frac{1}{K}\sum_{k=1}^{K} \mathbb{E}\left[\text{Tr}\left(P_k\right)\right],
\end{equation}
   where $\limsup_{K\rightarrow \infty}$ is the limit superior operator. 

The sensor's decision policy of transmission and retransmission is defined as
$
\pi \triangleq (a_1,a_2,...,a_k,\cdots).
$

In what follows, we optimize the sensor's transmission policy such that the long-term estimation error is minimized,~i.e.,
\begin{equation} \label{problem}
\min_{\pi}\limsup_{K\to\infty}\frac{1}{K}\sum_{k=1}^{K} \mathbb{E}\left[\text{Tr}\left(P_k\right)\right].
\end{equation}

%
		

\section{Performance-Optimal Policy}  \label{sec:performance_optimal}
   
   \subsection{MDP Formulation}     \label{sec:MDP}  
From \eqref{retransmission time}, \eqref{def_q} and \eqref{general_form}, the estimation MSE $\text{Tr}\left(P_k\right)$ and also the states $r_k$ and $q_k$ only depend on the current action $a_{k}$ and the previous states $r_{k-1}$ and $q_{k-1}$.
Thus, problem \eqref{problem} can be formulated as a discrete time Markov decision process (MDP) as follows. 

1) The state space is defined as 
$\mathbb{S} \triangleq \{(r, q) : r \leq q,\ (r, q) \in \mathbb{N}_0 \times \mathbb{N}_0\}$, where $\mathbb{N}_0$ is the set of non-negative integers, and the current retransmission time $r$ should be no larger than $q$ from the definition~\eqref{def_q_0}.
The state of the MDP at time $k$ is 
$s_k \triangleq (r_k, q_k) \in \mathbb{S}$. 
%

2) The action space is defined as $\mathbb{A} \triangleq \{0,1\}$. Recall that the action at time $k$, $a_k \in \mathbb{A}$, indicates a new transmission $(a_k=0)$ or a retransmission $(a_k=1)$.

3) The state transition function $P(s'|s,a)$ characterizes the probability that the state transits from state $s$ at time $(k-1)$ to $s'$ at time $k$ with action $a$ at time $k$. As the transition is time-homogeneous and the successful packet detection rate only depends on the number of retransmissions $r$, we can drop the time index $k$ here. Let $s=(r,q)$ and $s'=(r',q')$ denote the current and next state, respectively. 
Based on the HARQ successful packet detection probability~\eqref{HARQ_prob} and the iterations~\eqref{retransmission time} and \eqref{def_q}, we have the following state transition.
If the action $a=0$, the next state is 
\begin{equation}\label{s'_1}
s'=
\begin{cases}
(0,0), \mbox{ with probability } (1-g(0)) \\
(0,q+1), \mbox{ with probability } g(0).
\end{cases}
\end{equation}
If the action $a=1$, the next state is 
\begin{equation}\label{s'_2}
\hspace{-0.2cm}
s'=\begin{cases}
(r+1,r+1), \mbox{ with probability } (1-g(r+1))\\
(r+1, q+1), \mbox{ with probability } g(r+1).
\end{cases}
\end{equation}

4) The one-stage (instantaneous) cost based on \eqref{general_form} and \eqref{cost_function} is a function of the current state, which is independent of action:
\begin{equation} \label{one-stage cost}
c((r,q),a)\triangleq \text{Tr}\left(f^{q+1}(\bar{P}_0)\right).
\end{equation}

Since the cost function grows exponentially with the state~$q$, it is possible that the long-term average cost with a HARQ-based policy in the state space $\mathbb{S}$ cannot be bounded, i.e., the remote estimation system is unstable. We give the following sufficient condition of the existence of an optimal policy that has a bounded long-term MSE.



 \begin{theorem} \label{theorem:existence}
 	\normalfont
    	There exists a stationary and deterministic optimal policy $\pi^*$ of problem \eqref{problem} in the state space $\mathbb{S}$, if the following condition holds:
    	\begin{equation} \label{stability_condition}
    	(1-\lambda') \rho^2(A) <1, \text{ where } (1-\lambda')\triangleq \max\limits_{r>0}\{g(r)\}. 
    	\end{equation}
 \end{theorem}
\begin{proof}
	See Appendix A.
\end{proof}


\begin{remark}
	From Theorem~\ref{theorem:existence}, it is clear that the optimal policy exists if the channel condition is good (i.e., a smaller $g(r)$ and a smaller $1-\lambda'$) and the dynamic process does not change quickly (i.e., a small $\rho^2(A)$).
	Assuming the existence of a stationary and deterministic optimal policy, we can effectively solve the MDP problem using standard methods such as the relative value iteration algorithm~\cite[Chapter~8]{puterman2014markov}.  
\end{remark}
     
     \subsection{Structural Property of the Optimal Policy}
     The switching structure of the optimal policy is given as follows.
	\begin{theorem} \label{theorem:switching}
 	\normalfont		
		The optimal policy $\pi^*$ of problem \eqref{problem} is a switching-type policy, i.e.,		
(i) if $\pi^{*}(r,q)=0$, then  $\pi^{*}(r+z,q)=0$;
(ii) if $\pi^{*}(r,q)=1$, then  $\pi^{*}(r,q+z)=1$, where $z$ is any positive integer.
	\end{theorem}
	\begin{proof}
		See Appendix B.
	\end{proof}
In other words, for the optimal policy, the two-dimensional state space $\mathbb{S}$ is divided into two regions by a curve, and the decision actions of the states within each region are the same, which will be illustrated in Sec.~\ref{sec:num}.

\begin{remark}
	Note that the switching structure can help saving storage space for on-line implementation, since the smart sensor only needs to store switching-boundary states rather than the actions on the entire state space. At each time, the sensor simply needs to compare the current state with the boundary states to give the optimal decision.
\end{remark}

     \subsection{Suboptimal Policy}
     The optimal policy of the MDP problem does not have a closed-form expression for low-complexity computation.
	 Besides, since the MDP problem has infinitely many states,
     it has to be approximated by a truncated MDP problem with finite states for numerical evaluation and solved offline.
     Therefore, we propose a easy-to-compute suboptimal policy, which is the myopic policy that makes decision simply to maximize the expected next step cost.
     
     Based on \eqref{s'_1}, \eqref{s'_2} and \eqref{one-stage cost}, the expected next step cost $c'((r,q),a)$ given the current state $(r,q)$ can be derived~as
     \begin{equation}
     \begin{aligned}
     &c'((r,q),a) \\
     &=\!\begin{cases}
     \!g(0)\text{Tr}\left(f^{q+2}(\bar{P}_0)\right) \!+\! (1-g(0))\text{Tr}\left(f(\bar{P}_0)\right), \mbox{ if } a = 0;\\
     \!g(r+1)\text{Tr}\left(\!f^{q+2}(\bar{P}_0)\right) \!+\! (1-g(r+1))\text{Tr}\left(\!f^{r+2}(\bar{P}_0)\right) \\
     \hspace{6.3cm}\mbox{ if } a = 1.
     \end{cases}
     \end{aligned}
     \end{equation}
     Then, we have
     \begin{equation}
\begin{aligned}
     &c'((r,q),1)-c'((r,q),0) \\
     &= (g(r+1)-g(0))\text{Tr}\left(f^{q+2}(\bar{P}_0)\right)\\
     &+(1-g(r+1))\text{Tr}\left(f^{r+2}(\bar{P}_0)\right)-(1-g(0))\text{Tr}\left(f(\bar{P}_0)\right).
\end{aligned}
     \end{equation}
      Since $g(0)>g(r)$ when $r>0$, $c'((r,q),1)-c'((r,q),0) \geq 0$ if and only if $(r,q)$ satisfies
      \begin{equation} \label{condition}
\begin{aligned}
      &\text{Tr}\left(f^{q+2}(\bar{P}_0)\right) \\
      &\leq \frac{(1-g(r+1))\text{Tr}\left(f^{r+2}(\bar{P}_0)\right) -(1-g(0))\text{Tr}\left(f(\bar{P}_0)\right)}{g(0)-g(r+1)}.
\end{aligned}
      \end{equation}
Thus, we have the following result.      
      \begin{proposition} \label{prop:sub}
 	\normalfont      	
      	A suboptimal policy of problem~\eqref{problem} is
      	\begin{equation} \label{sub_policy}
      a = \begin{cases}
      	0 &\mbox{if the condition \eqref{condition} is satisfied,}\\
      	1 &\mbox{otherwise.}\\
      	\end{cases}
      	\end{equation}
      \end{proposition}

     It can be proved that the suboptimal policy in Proposition~\ref{prop:sub} is also a switching-type policy. Moreover, based on \eqref{sub_policy} and the monotonicity of  $\text{Tr}\left(f^{n}(\bar{P}_0)\right)$ w.r.t. $n$ discussed in Sec.~\ref{sec:HARQ}, it can be verified that the action should always be zero for the states $(r,q) \in \mathbb{S}$ with $r=q$, i.e., a new transmission is required.
\emph{Due to the simplicity of the suboptimal policy, which, unlike the optimal policy, does not need any iteration for policy calculation, it can be applied as an on-line decision algorithm. } In Sec.~\ref{sec:num}, we will show that the performance of the suboptimal policy is close to the optimal one for practical system parameters. The detailed computing complexity analysis of the policies is omitted due to the space limitation.
	      
\section{Delay-Optimal Policy: A Benchmark} \label{sec:delay_optimal}
    We also consider a delay-optimal policy based on the HARQ protocol, which is similar to \cite{ceran2018average}, as the benchmark of the proposed performance-optimal policy.
    
    We use the AoI to measure the delay of the system. Specifically, $\tau_k$ is the AoI of the system at the beginning of time slot $k$.
    Due to the definition of $q_k$ in \eqref{def_q_0}, it is clear that 
$
    \tau_k = k-t_{k-1} = q_{k-1}+1.
$
   Therefore, similar to the performance optimization problem~\eqref{problem}, the delay optimization problem is formulated as
  $
    \min_{\pi}\limsup_{K\to\infty}\frac{1}{K}\sum_{k=1}^{K} \mathbb{E}\left[\tau_k\right].
    $
	This~problem can also be converted to a MDP problem with the same state space, action space and state transition function~as presented in Sec.~\ref{sec:MDP}. The one-stage cost in terms of delay~is 
	\begin{equation} \label{one-stage cost2}
		c((r,q),a) = q+1.
	\end{equation}
    
    Comparing \eqref{one-stage cost2} with \eqref{one-stage cost}, we see that the cost function of the delay-optimal policy is a linear function of $q$, while it grows exponentially fast with $q$ in the performance-optimal policy. Thus, these two policies should be different and their performance will be compared in the following section.
     
   \section{Numerical Results} \label{sec:num}
In this section, we present numerical results of the optimal policy in Sec.~\ref{sec:performance_optimal} and its performance. Also, we numerically compare the performance-optimal policy with the benchmark policy in Sec.~\ref{sec:delay_optimal}.
Unless otherwise stated, we set
$A = \begin{bmatrix}
1.8 & 0.2 \\
0.2 & 0.8
\end{bmatrix}$, $C = \begin{bmatrix}
1 & 1
\end{bmatrix}$, $Q = I$, $R = 1$, 
and thus
$\rho^2(A) = 1.8385^2$, 
$\bar{P}_0 = \begin{bmatrix}
2.3579  & -1.5419 \\
-1.5419 &  1.5987
\end{bmatrix}$. 
The successful detection probability of a new transmission is $\lambda = 0.8$. 

Due to the exponential behavior of the error probability of HARQ~\cite{frenger2001performance,tripathi2003reliability},
the packet detection error probability of a HARQ protocol is approximated as $g(r) = (1-\lambda) h^r $ for $r\geq 0$. It can be verified that condition~\eqref{stability_condition} holds, i.e., the optimal policy exits.
The parameter $h$ is determined by the HARQ combining scheme (e.g., the incremental redundancy scheme has a smaller $h$, i.e., a better performance, than the chase combining scheme). 

\emph{\underline{Policy Comparison.}}
We use the relative value iteration algorithm based on the Matlab MDP toolbox to solve the MDP problems in Sections~\ref{sec:performance_optimal} and \ref{sec:delay_optimal},
where the unbounded state space $\mathbb{S}$ is truncated as $\{(r, q) : 0 \leq r\leq q\leq 20\}$ to enable the evaluation. 
Fig.~\ref{policy_comparison} shows different policies with different parameter $h$ within the truncated state space. 
In Fig.~\ref{policy_comparison}(a), we see that in line with Theorem~\ref{theorem:switching}, the optimal policy is a switching-type one, where the actions of the states that are close to the states with $r=q$, are equal to zero, i.e., new transmissions are required.
Also, we see that the suboptimal policy plotted in Fig.~\ref{policy_comparison}(b) is a good approximation of the optimal one within the truncated state space. However, the delay-optimal policy plotted in Fig.~\ref{policy_comparison}(c) is very different from the previous ones, where more states have the action of new transmission.
Therefore, retransmissions are more important to reduce the estimation MSE than the delay.
Fig.~\ref{policy_comparison}(d) presents the optimal policy with $h = 0.9$. Comparing with Fig.~\ref{policy_comparison}(a), we see that more states have to choose the action of new transmission with the HARQ protocol having a larger $h$, i.e., a worse HARQ combining scheme.


\begin{figure}[t]
	\centering
	\includegraphics[scale=0.5]{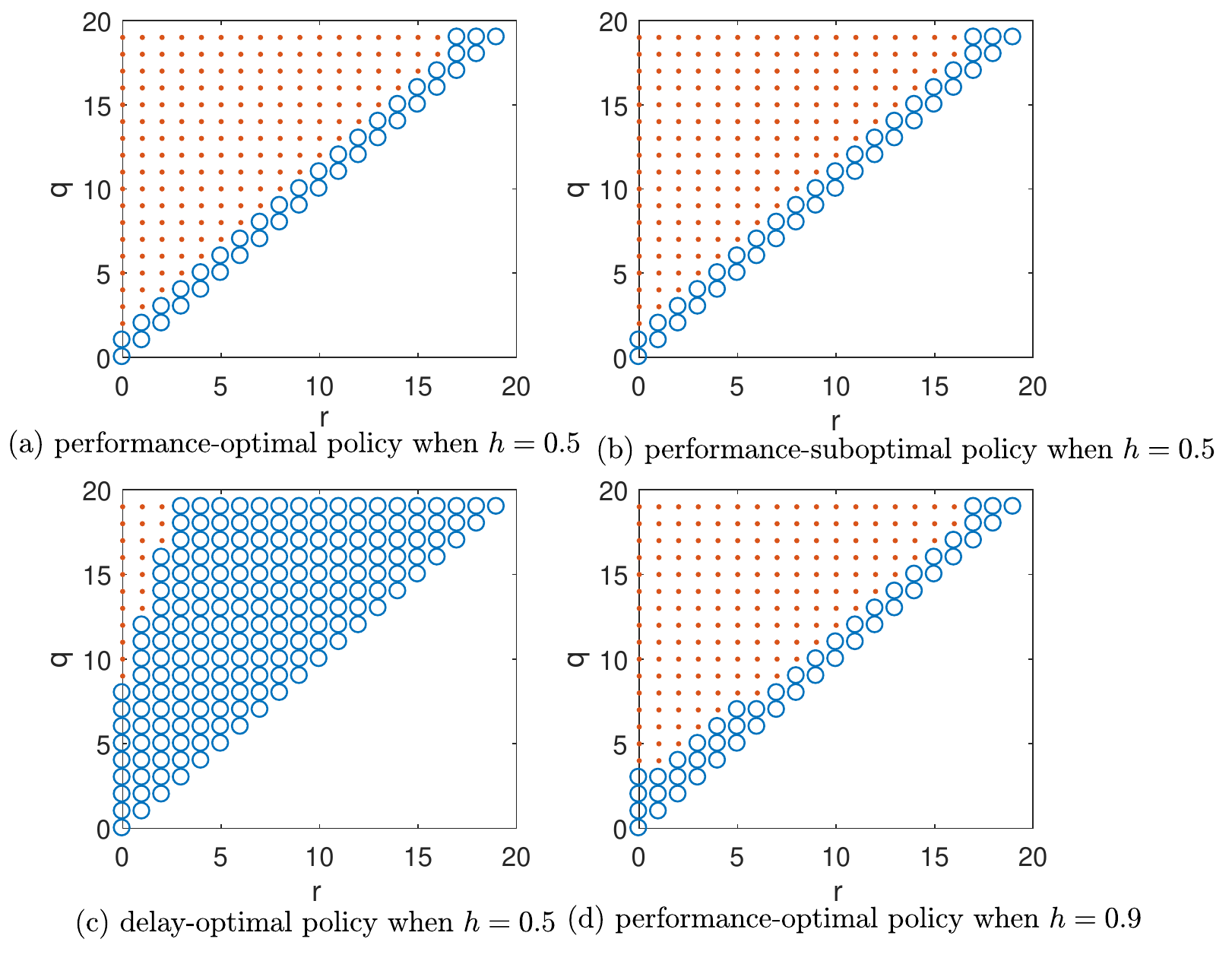}
	\vspace{-0.3cm}
	\caption{An illustration of different policies with different $h$, where `o' and `$\cdot$' denote $a =0$ and $a=1$, respectively.}
	\label{policy_comparison}
     	\vspace{-0.5cm}
\end{figure}

\emph{\underline{Performance Comparison.}}
Based on the above numerically obtained polices and the policy with the standard ARQ, i.e., the one without retransmission (see Sec.~\ref{sec:ARQ}), we further evaluate their performances in terms of the long-term average MSE using \eqref{cost_function}.
We run $2000$ Monte Carlo simulations with the initial value of $P_k$ as $P_0 = f(\bar{P}_0) = \begin{bmatrix}
7.5934  & -1.1774\\
-1.1774  &  1.6241
\end{bmatrix}$.
Also, we set $\text{Tr}(P_0) =~9.2$ as the \emph{performance baseline}, as $\text{Tr}(P_0)\leq \text{Tr}(P_k)$, $\forall k$.

Fig.~\ref{performance_0.1} plots the average MSE versus the simulation time $K$, using different policies with $h=0.5$.
We see that the average MSEs of different policies converge to the steady state values when $K > 1200$.
Given the performance baseline, the performance-optimal policy gives a $32\%$ and $10\%$ MSE reduction of the non-retransmission policy when $\lambda = 0.8$ and $0.85$, respectively.
This shows that the performance improvement by the HARQ-based policy is more significant when we have a worse channel quality.
The performance gap between the performance- and delay-optimal policies in terms of MSE is noticeable for these cases, which demonstrates the superior of the proposed optimal one.

\begin{figure}[t]
	\centering
	\includegraphics[scale=0.5]{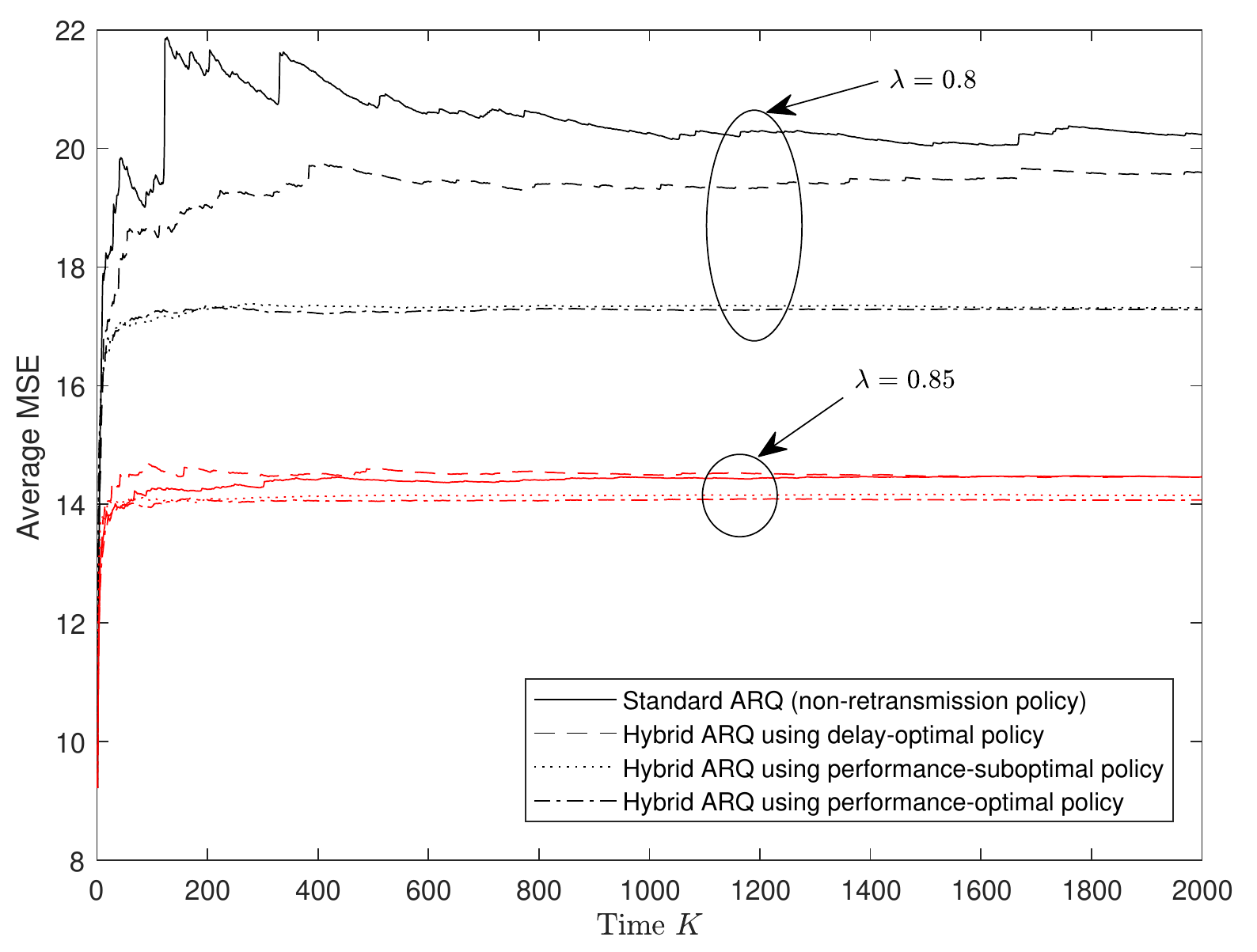}
	\vspace{-0.3cm}
	\caption{Average MSE with different policies, $h = 0.5$}
	\label{performance_0.1}
	     	\vspace{-0.5cm}
\end{figure}

\section{Conclusions}
We have proposed and optimized a HARQ-based remote estimation protocol for real-time applications.
Our results have shown that the optimal policy is able to achieve a remarkable $30\%$ estimation MSE reduction for some practical settings.
As the recent
communication standards for real-time wireless control, such
as WirelessHART, ISA-100 and IEEE 802.15.4e, have not
adopted any HARQ techniques,
this work also suggests that HARQ can be adopted by the future real-time communication standards  to enhance the system performance.

     \section*{Appendix A: Proof of Theorem~\ref{theorem:existence}}
	To prove the existence of a stationary and deterministic optimal policy given condition~\eqref{stability_condition}, we need to verify the following conditions~\cite[Corollary 7.5.10]{sennott2009stochastic}:
	(CAV*1) there exists a standard policy $\psi$ such that the recurrent class $R_{\psi}$ induced by $\psi$ is equal to the whole state space $\mathbb{S}$;
    (CAV*2) given $U > 0$, the set $\mathbb{S}_U = \{s|c(s,a) \leq U  \mbox{ for some } a\}$ is finite.	    
   
    Condition (CAV*2) can be easily verified based on \eqref{one-stage cost}. In what follows, we verify (CAV*1) by first constructing a policy $\psi$ and then proving that it is a standard policy.
    
    The action of the policy $\psi$ is given as 
    \begin{equation} \label{standard}
    a = \psi(s) = \psi(r,q)
    =
    \begin{cases}
    0, \ r=q\\
    1, \ \text{otherwise}.
    \end{cases}
    \end{equation}
    It is easy to prove that any state in $\mathbb{S}$ induced by $\psi$ is a recurrent state.
    We then prove that $\psi$ is a standard policy by verifying both the expected first passage cost and time from state $(r,q) \in \mathbb{S}\backslash (0,0)$ to $(0,0)$ are bounded~\cite{sennott2009stochastic}. Due to the space limitation, we only prove that any state with $r=q$ has bounded first passage cost and time. The other states can be proved similarly. 
    
    For simplicity, the expected first passage cost of the state $(i,i)$ is denoted as $d(i)$, and the one-stage cost \eqref{one-stage cost} is rewritten as
    $
    c(q) \triangleq c((r,q),a) = \text{Tr}\left(f^{q+1}(\bar{P}_0)\right).
    $
    Based on \eqref{HARQ_prob}, \eqref{standard} and the law of total expectation, we have
    \begin{equation} \label{equation_group}
\begin{aligned}
    &d(i) = c(i) + (1-g(0))c(0) + g(0) c(i+1)\\
    &+ g(0) (1-g(1)) d(1) + g(0) g(1) c(i+2)\\
    &+ g(0) g(1) (1-g(2)) d(2) + g(0) g(1) g(2) c(i+3) + \cdots\\
    &= \nu(i) +(1-g(0))c(0)+ D, \forall i >0,
\end{aligned}
    \end{equation}
    where $g(0) = 1-\lambda$,    
\begin{align}
\label{h}
    &\nu(i) = c(i) +\sum_{j=1}^{\infty} \alpha_j c(i+j),\ D= \sum_{j=1}^{\infty} \beta_j d(j),
\end{align}
    and $\alpha_j = \prod_{l=1}^{j} g(l-1)$ and $\beta_j = \prod_{l=1}^{j} g(l-1)(1-g(j))$.
    Therefore, $d(i)$ is bounded if $\nu(i)< \infty$ and $D< \infty$. Since $g(r)\leq (1-\lambda')$ when $r>0$, we have $\alpha_j \leq (1-\lambda)\left(1-\lambda'\right)^{j-1}$.
    From~\cite{schenato2008optimal}, we have $\sum_{j=1}^{\infty} (1-\lambda')^jc(j) < \infty$ iff $(1-\lambda') \rho^2(A)<1$.
    Thus, it is easy to prove that $\nu(i) < \infty$ if \eqref{stability_condition} holds.

    From \eqref{equation_group}, $D$ can be further derived after simplifications as
    \begin{equation}
    D = \frac{1}{1-\sum_{i=1}^{\infty}\beta_i} \left(\sum_{i=1}^{\infty} \beta_i (1-g(0))c(0) + \sum_{i=1}^{\infty} \beta_i \nu(i) \right).
    \end{equation}
    As $\sum_{i=1}^{\infty}\! \beta_i \!=\!g(0)\!<\! 1$, $D$ is bounded as long as $\sum_{i=1}^{\infty} \!\beta_i \nu(i) \!<\! \infty$. Since $\alpha_i$, $\beta_i \leq (1-\lambda)(1-\lambda')^{i-1}$, after some simplifications, we have
    \begin{equation}
    \begin{aligned}\label{temp}
    \sum_{i=1}^{\infty} \beta_i \nu(i) 
    \leq 
    \eta \sum_{j=1}^{\infty} (1-\lambda')^{j} c(j)+
    \eta^2\sum_{j=2}^{\infty}  (j-1) (1-\lambda')^{j} c(j),
    \end{aligned}
    \end{equation}
    where $\eta = (1-\lambda')/(1-\lambda)$.
    It can be proved that $\sum_{j=2}^{\infty}  (j-1) (1-\lambda')^{j} c(j)$ is bounded if $\sum_{j=1}^{\infty} (1-\lambda')^{j} c(j)$ is bounded. Again, using the result that 
    $\sum_{j=1}^{\infty} (1-\lambda')^jc(j) < \infty$ iff $(1-\lambda') \rho^2(A)<1$ in \cite{schenato2008optimal}, $\sum_{i=1}^{\infty} \beta_i \nu(i)<\infty$ if $(1-\lambda') \rho^2(A)<1$, yielding the proof of the bounded expected first passage cost with condition~\eqref{stability_condition}. Similarly, we can verify that the expected first passage time
    is also bounded.

    \section*{Appendix B: Proof of Theorem~\ref{theorem:switching}}
    	The switching property is equivalent to the monotonicity of the optimal policy in $r$ if $q$ is fixed and in $q$ if $r$ is fixed. The monotonicity can be proved by verifying the following conditions (see Theorem 8.11.3 in \cite{puterman2014markov}).
    
    (1) $c(s,a)$ is nondecreasing in $s$ for all $a \in \mathbb{A}$;
    
    (2) $c(s,a)$ is a superadditive function on $\mathbb{S} \times \mathbb{A}$;
    
    (3) $q(s'|s,a)= \sum_{i=s'}^{\infty}\mathbb{P}\left[i|s,a\right]$ is nondecreasing in $s$ for all $s'\in \mathbb{S}$ and $a \in \mathbb{A}$;
    
    (4)  $q(s'|s,a)$ is a superadditive function on $\mathbb{S} \times \mathbb{A}$ for all  $s'\in \mathbb{S}$.
    
    We first prove the monotonicity in $r$ with $q$ fixed. The state $s$ is ordered by $r$, i.e., if $r^-\leq r^+$, we define $s^- \leq s^+$ with $s^-=(r^-,q)$ and $s^+=(r^+,q)$. 
	From the definition of one-stage cost, $c(s,a)$ is increasing in $q$. Therefore, condition (1) can be easily verified.        
    For condition (2), the superadditive function is defined in (4.7.1) of \cite{puterman2014markov}. A function $f(x,y)$ is superadditive for $x^{-} \leq x^{+}$ and $y^{-} \leq y^{+}$, if
    $f(x^{+},y^{+})+f(x^{-},y^{-}) \geq f(x^{+},y^{-})+f(x^{-},y^{+})$. Then, condition (2) can be easily verified as $c(s,a)$ is independent of $a$.
    
    Given the current state $s = (r,q)$, from \eqref{s'_1} and \eqref{s'_2}, the next possible states are $s_0 \triangleq (0,0)$, $s_1 \triangleq (0,q+1)$, $s_2 \triangleq (r+1,r+1)$ and $s_3 \triangleq (r+1,q+1)$.
   Let $s'\triangleq\left\lbrace(r',q'): q\in \mathbb{N}_0 \right\rbrace$.
    If $r'\leq r$, we define $s' \preceq s$ with $s=(r,q)$. 
    Based on \eqref{s'_1} and \eqref{s'_2}, $q(s'|s,a)$ with different actions are given as:\\
$
q(s'|s,a\!=\!0)\!=\!\begin{cases}
1, &\mbox{if } s' \preceq s_0\\
0, &\mbox{otherwise }
\end{cases}
$,
and
$
q(s'|s,a\!=\!1)\!=\!\begin{cases}
1, &\mbox{if } s' \preceq s_2\\
0, &\mbox{otherwise}
\end{cases}.
$
%
    Therefore, condition (3) can be easily verified.
    
    For condition (4),    
    let $s^{+} = (r^{+},q)$, $s^{-} = (r^{-},q)$,  $r^{+} \geq r^{-}$ and  $a^{+} \geq a^{-}$ Then, we need to verify if
    $
    q(s'|s^{+},a^{+})+q(s'|s^{-},a^{-}) \geq q(s'|s^{+},a^{-})+q(s'|s^{-},a^{+}).
    $
    Based on the definitions of $q(s'|s,a)$, $s'$ and $s_i$, $i=0,1,2,3$, 
    condition (4) can be verified straightforwardly.
	As all four conditions hold, the monotonicity of the optimal policy in $r$ is proved. Similarly, the monotonicity of the optimal policy in $q$ can be proved. 

    \balance
    
	\ifCLASSOPTIONcaptionsoff
	\newpage
	\fi


\begin{thebibliography}{10}
	\providecommand{\url}[1]{#1}
	\csname url@samestyle\endcsname
	\providecommand{\newblock}{\relax}
	\providecommand{\bibinfo}[2]{#2}
	\providecommand{\BIBentrySTDinterwordspacing}{\spaceskip=0pt\relax}
	\providecommand{\BIBentryALTinterwordstretchfactor}{4}
	\providecommand{\BIBentryALTinterwordspacing}{\spaceskip=\fontdimen2\font plus
		\BIBentryALTinterwordstretchfactor\fontdimen3\font minus
		\fontdimen4\font\relax}
	\providecommand{\BIBforeignlanguage}[2]{{%
			\expandafter\ifx\csname l@#1\endcsname\relax
			\typeout{** WARNING: IEEEtran.bst: No hyphenation pattern has been}%
			\typeout{** loaded for the language `#1'. Using the pattern for}%
			\typeout{** the default language instead.}%
			\else
			\language=\csname l@#1\endcsname
			\fi
			#2}}
	\providecommand{\BIBdecl}{\relax}
	\BIBdecl
	
	\bibitem{Antonakoglou2018towards}
	K.~Antonakoglou, X.~Xu, E.~Steinbach, T.~Mahmoodi, and M.~Dohler, ``Towards
	haptic communications over the {5G} tactile {Internet},'' \emph{to appear in
		IEEE Commun. Surveys Tuts.}, 2018.
	
	\bibitem{kaul2012real}
	S.~Kaul, R.~Yates, and M.~Gruteser, ``Real-time status: How often should one
	update?'' in \emph{Proc. IEEE INFOCOM}, Mar. 2012, pp. 2731--2735.
	
	\bibitem{schenato2008optimal}
	L.~Schenato, ``Optimal estimation in networked control systems subject to
	random delay and packet drop,'' \emph{IEEE Trans. Autom. Control}, vol.~53,
	no.~5, pp. 1311--1317, Jun. 2008.
	
	\bibitem{sun2017remote}
	Y.~Sun, Y.~Polyanskiy, and E.~Uysal-Biyikoglu, ``Remote estimation of the
	wiener process over a channel with random delay,'' in \emph{Proc. IEEE ISIT},
	Jun. 2017, pp. 321--325.
	
	\bibitem{liu2016energy}
	W.~Liu, X.~Zhou, S.~Durrani, H.~Mehrpouyan, and S.~D. Blostein, ``Energy
	harvesting wireless sensor networks: Delay analysis considering energy costs
	of sensing and transmission,'' \emph{IEEE Trans. Wireless Commun.}, vol.~15,
	no.~7, pp. 4635--4650, Jul. 2016.
	
	\bibitem{Demirel2015to}
	B.~Demirel, A.~Aytekin, D.~E. Quevedo, and M.~Johansson, ``To wait or to drop:
	On the optimal number of retransmissions in wireless control,'' in
	\emph{Proc. ECC}, Jul. 2015, pp. 962--968.
	
	\bibitem{gupta2010estimation}
	V.~Gupta, ``On estimation across analog erasure links with and without
	acknowledgements,'' \emph{IEEE Trans. Autom. Control}, vol.~55, no.~12, pp.
	2896--2901, Dec. 2010.
	
	\bibitem{caire2001throughput}
	G.~Caire and D.~Tuninetti, ``The throughput of hybrid-{ARQ} protocols for the
	{Gaussian} collision channel,'' \emph{IEEE Trans. Inf. Theory}, vol.~47,
	no.~5, pp. 1971--1988, Jul. 2001.
	
	\bibitem{yang2015deterministic}
	C.~Yang, J.~Wu, X.~Ren, W.~Yang, H.~Shi, and L.~Shi, ``Deterministic sensor
	selection for centralized state estimation under limited communication
	resource,'' \emph{IEEE Trans. Signal Process.}, vol.~63, no.~9, pp.
	2336--2348, May 2015.
	
	\bibitem{shi2012optimal}
	L.~Shi and L.~Xie, ``Optimal sensor power scheduling for state estimation of
	{Gauss-Markov} systems over a packet-dropping network,'' \emph{IEEE Trans.
		Signal Process.}, vol.~60, no.~5, pp. 2701--2705, May 2012.
	
	\bibitem{yang2013schedule}
	C.~Yang, J.~Wu, W.~Zhang, and L.~Shi, ``Schedule communication for
	decentralized state estimation,'' \emph{IEEE Trans. Signal Process.},
	vol.~61, no.~10, pp. 2525--2535, May 2013.
	
	\bibitem{maybeck1979stochastic}
	P.~S. Maybeck, \emph{Stochastic models, estimation, and control}.\hskip 1em
	plus 0.5em minus 0.4em\relax Academic press, 1982, vol.~3.
	
	\bibitem{shi2012scheduling}
	L.~Shi and H.~Zhang, ``Scheduling two {Gauss-Markov} systems: An optimal
	solution for remote state estimation under bandwidth constraint,'' \emph{IEEE
		Trans. Signal Process.}, vol.~60, no.~4, pp. 2038--2042, Apr. 2012.
	
	\bibitem{frenger2001performance}
	P.~Frenger, S.~Parkvall, and E.~Dahlman, ``Performance comparison of {HARQ}
	with chase combining and incremental redundancy for {HSDPA},'' in \emph{Proc.
		IEEE VTC}, vol.~3, Oct. 2001, pp. 1829--1833.
	
	\bibitem{tripathi2003reliability}
	V.~Tripathi, E.~Visotsky, R.~Peterson, and M.~Honig, ``Reliability-based type
	{II} hybrid {ARQ} schemes,'' in \emph{Proc. IEEE ICC}, Jun. 2003, pp.
	2899--2903.
	
	\bibitem{puterman2014markov}
	M.~L. Puterman, \emph{Markov decision processes: discrete stochastic dynamic
		programming}.\hskip 1em plus 0.5em minus 0.4em\relax John Wiley \& Sons,
	2014.
	
	\bibitem{ceran2018average}
	E.~T. Ceran, D.~G{\"u}nd{\"u}z, and A.~Gy{\"o}rgy, ``Average age of information
	with hybrid {ARQ} under a resource constraint,'' in \emph{Proc. IEEE WCNC},
	Apr. 2018, pp. 1--6.
	
	\bibitem{sennott2009stochastic}
	L.~I. Sennott, \emph{Stochastic dynamic programming and the control of queueing
		systems}.\hskip 1em plus 0.5em minus 0.4em\relax John Wiley \& Sons, 2009,
	vol. 504.
	
\end{thebibliography}

%
\bibliographystyle{IEEEtran}

\end{document}